\documentclass[12pt,a4paper]{article}

\usepackage{amsmath,amssymb,amsfonts,amsthm,latexsym,color,graphicx,ulem}

\definecolor{myred}{rgb}{1,0.1,0.2}

\newtheorem{thm}{Theorem}[section]

\newcommand{\Ind}{\mathbf{1}}
\newcommand{\E}{\operatorname{E}}

\newcommand{\Var}{\operatorname{Var}}
\newcommand{\trace}{\operatorname{trace}}
\renewcommand{\P}{\operatorname{P}}
\renewcommand{\mathbf}{\boldsymbol}

\newcommand{\argmin}{\mathop{\mathrm{argmin}}}

\renewcommand{\bf}{\bfseries}
\renewcommand{\tilde}[1]{\widetilde{#1}}
\renewcommand{\hat}[1]{\widehat{#1}}

\begin{document}

\title{Robust Model Selection in Generalized Linear Models
}
\author{Samuel M\"uller\\
        School of Mathematics and Statistics \\
        University of Western Australia, Crawley, WA 6009, Australia\\
        \&\\
        A. H.\ Welsh \\
        Centre for Mathematics and its Applications \\
        Australian National University, Canberra, ACT 0200, Australia }

\date{\today}

\maketitle

\abstract{In this paper, we extend to generalized linear models (including logistic and other binary regression models, Poisson regression and gamma regression models) the robust model selection methodology developed by M\"uller and Welsh (2005) for linear regression models.  As in M\"uller and Welsh (2005), we combine a robust penalized measure of fit to the sample with a robust measure of out of sample predictive ability which is estimated using a post-stratified m-out-of-n bootstrap.  A key idea is that the method can be used to compare different estimators (robust and nonrobust) as well as different models. Even when specialized back to linear regression models, the methodology presented in this paper improves on that of M\"uller and Welsh (2005).  In particular, we use a new bias-adjusted bootstrap estimator which avoids the need to centre the explanatory variables and to include an intercept in every model.  We also use more sophisticated arguments than M\"uller and Welsh (2005) to establish an essential monotonicity condition.}

\bigskip
{\it Keywords:} bootstrap model selection, generalized linear models, paired bootstrap, robust estimation, robust model selection, stratified bootstrap

\thispagestyle{empty}

\newpage
\pagenumbering{arabic}
\setcounter{page}{1}

\section{Introduction} \label{sec: Introduction}
Model selection is fundamental to the practical application of statistics and there is a substantial literature on the selection of linear regression models.  A growing part of this literature is concerned with robust approaches to selecting linear regression models: see M\"uller and Welsh (2005) for references. The literature on the selection of generalized linear models (GLM; McCullagh and Nelder, 1989) and the related marginal models fitted by generalized estimating equations (GEE; Liang and Zeger, 1986) -- though both are widely used -- is much smaller and has only recently incorporated robustness considerations.  Hurvich and Tsai (1995) and Pan (2001) developed Akaike information criterion (AIC) like criteria based on the quasi-likelihood, Cantoni, Mills Flemming, and Ronchetti (2005) presented a generalized version of Mallows' $C_{p}$, and Pan and Le (2001) and Cantoni et al.\ (2007) presented approaches based on the bootstrap and cross-validation, respectively.  Our purpose in this paper is to generalize the robust bootstrap model selection criterion of M\"uller and Welsh (2005) to generalized linear models.

The extension of the methodology of M\"uller and Welsh (2005) from linear regression to generalized linear models is less straightforward than we expected and, as a result, the present paper differs from M\"uller and Welsh (2005) in two important respects.  First, the bias-adjusted $m$-out-of-$n$ bootstrap estimator $\hat{\beta}^{c*}_{\alpha,m} - \E_*(\hat{\beta}^{c*}_{\alpha,m} - \hat{\beta}_{\alpha}^c)$ rather than the $m$-out-of-$n$ bootstrap estimator $\hat{\beta}^{c*}_{\alpha,m}$ is used in estimating the expected prediction loss $M_n^{(2)}(\alpha)$ (definitions are given in Section 2).  As discussed in more detail in Section 3.2, this achieves the same purpose but avoids the centering of the explanatory variables and the requirement that we include an intercept in every model used in M\"uller and Welsh (2005).  Second, we present a simpler, more general method than that used in M\"uller and Welsh (2005) for showing that the consistency result applies to particular robust estimators of the regression parameter.  As discussed in Section 3.3, we use generalized inverse matrices to decompose the asymptotic variance of the estimator into terms which are easier to handle, write the critical trace term as a simple sum and then show that the terms in this sum have the required properties.  Both of these changes were necessitated by the more complicated structure of generalized linear models but they also apply to regression models where they represent improvements to the methodology of M\"uller and Welsh (2005).

Suppose that we have $n$ independent observations $y=(y_1,\ldots,y_n)^T$ and an $n \times p$ matrix $X$ whose columns we index by $\{1,\ldots,p\}$.  
Let $\alpha$ denote any subset of $p_{\alpha}$ distinct elements from $\{1,\ldots,p\}$ and let $X_{\alpha}$ denote the $n \times p_{\alpha}$ matrix with columns given by the columns of $X$ whose indices appear in $\alpha$.  Let $x_{\alpha i}^T$ denote the $i$th row of $X_{\alpha}$.  
Then a generalized linear regression model $\alpha$ for the relationship between the response variable $y$ and explanatory variables $X$ is specified by
\begin{equation}
\E y_{i} = h(\eta_i),   \,\,\,\, \Var y_{i} = \sigma^{2} v^{2}(\eta_{i}) \mbox{ with } \eta_{i} = x_{\alpha i}^T \beta_{\alpha}, \,\,\,\, i = 1,\ldots, n, \label{eqn:model}
\end{equation}
where $\beta_{\alpha}$ is an unknown $p_{\alpha}$-vector of regression parameters.  
Here $h$ is the inverse of the usual link function and, for simplicity, we have reduced notation by absorbing $h$ into the variance function $v$.  Both $h$ and $v$ are assumed known.
Let $\mathcal{A}$ denote a set of generalized linear regression models for the relationship between $y$ and $X$.   The purpose of model selection is to choose one or more models $\alpha$ from $\mathcal{A}$ with specified desirable properties.

Our perspective on model selection is that a useful model should 
(i) parsimoniously describe the relationship between the sample data $y$ and $X$ and
(ii) be able to predict independent new observations.
The ability to parsimoniously describe the relationship between the sample data can be measured by applying a penalised loss function to the observed residuals and we use the expected variance-weighted prediction loss to measure the ability to predict new observations.  In addition, we
encourage the consideration of different types of estimator of each of the models.  Possible estimators include the nonrobust maximum likelihood (see K\"unsch, Stefanski, and Carroll, 1989; Cantoni and Ronchetti, 2001; Ruckstuhl and Welsh, 2001) and the maximum quasi--likelihood estimators (see McCullagh and Nelder, 1989)  and the robust estimators of Preisser and Qaqish (1999), Cantoni and Ronchetti (2001), and Cantoni (2004).   The Cantoni and Ronchetti (2001) estimator is described in Section \ref{sec: examples and counterexamples}.

We define a class of robust model selection criteria in Section 2, present our theoretical results in Section 3, report the results of a simulation study in Section 4, present a real data example in Section 5, and conclude with a short discussion and some brief remarks in Section 6.

\section{Robust model selection criterion}

Let $\hat{\beta}_{\alpha}^c$ denote an estimator of type $c \in \mathcal{C}$ of $\beta_{\alpha}$ under (\ref{eqn:model}), let $\sigma$ be a scale parameter, let $\rho$ be a nonnegative loss function, let $\delta$ be a specified function of the sample size $n$, let $\sigma$ denote a measure of spread of the data, and let $\tilde y$ be a vector of future observations at $X$ which are independent of $y$.  Then, we choose models $\alpha$ from a set $\mathcal{A}$ for which the criterion function
\begin{eqnarray}
M(\alpha) & = &\frac{\sigma^{2}}{n} \bigg\{
               \E\sum_{i=1}^n w_{\alpha i}\rho[\{y_i-h(x_{\alpha i}^T\hat{\beta}_{\alpha}^c)\}/ \sigma v(\eta_{i})]+  \delta(n)p_\alpha \nonumber \\
               & &   + \E\bigg(\sum_{i=1}^n w_{\alpha i} \rho[\{\tilde y_i-h(x_{\alpha i}^T\hat{\beta}_{\alpha}^c)\}/ \sigma v(\eta_{i})] \,\Big|\, y, X \bigg)\bigg\} \label{eqn:criterion}
\end{eqnarray}
is small.  In practice, we often supplement this criterion with graphical diagnostic methods which further explore the quality of the model in ways that are not amenable to simple mathematical description.

As in M\"uller and Welsh (2005) we separate the estimators $\hat{\beta}_{\alpha}^c$ and $\rho$ because in practice we want to compare different estimators indexed by $c \in \mathcal{C}$ and linking $\rho$ to any one of these estimators may excessively favour that estimator.  We adopt the view that we are interested in fitting the core data and predicting core observations rather than those in the tail of the distribution so take $\rho$ to be constant for sufficiently large $|x|$.  The simplest example of such a function (and the one we use in our simulations) is the function which is quadratic near the origin and constant away from the origin as in
\begin{equation}
\rho(z) = \min(z^2, b^2). \label{eqn:rho}
\end{equation}
Following M\"uller and Welsh (2005), we choose $b=2$. Smoother versions of $\rho$ such as are required in our theoretical results are easily defined and we can, when appropriate to the problem, use asymmetric $\rho$ functions.    The weights $w_{\alpha i}$ are Mallows' type weights which may be included for robustness in the $X$ space but can and often will be constant.  The only restrictions on the function $\delta$ are that $\delta(n) \rightarrow \infty$ and $\delta(n)/n \rightarrow 0$ as $n \rightarrow \infty$.  A common choice is $\delta(n) = k \log(n)$ for $k>0$ where we choose $k=2$ (e.g. Schwarz, 1978; M\"uller and Welsh, 2005).  If the criterion were based on the penalized loss function alone then $\delta$ would have to be of order higher than $O(\log\log n)$ as shown in Qian and Field (2002, Theorem 1--3) for logistic regression models.

Let $\hat{\beta}^c_{\alpha}$ be an estimator of type $c$ of the model $\alpha$, and if $\sigma$ has to be estimated, we estimate it from the Pearson residuals $\{y_i-h(x_{\alpha_f i}^T\hat{\beta}^c_{\alpha_f})\}/v(x_{\alpha_f i}^T\hat{\beta}^c_{\alpha_f})$, $i=1,\dots,n$, from a ``full'' model $\alpha_f$.   A ``full''  model is a large model (often assumed to be the model $\{1,\ldots,p\}$) which produces a valid measure of residual spread (but hopefully not so large that we incur a high cost from overfitting).  We omit the subscript $\alpha_f$ and denote the estimator of $\sigma$ by $\hat{\sigma}$ for notational simplicity.   
Then we estimate the penalized in-sample term in the criterion function (\ref{eqn:criterion}) by $ \hat{\sigma}^{c2}\{M_n^{(1)}(\alpha)  +  n^{-1}\delta(n) p_{\alpha}\}$, where
\begin{eqnarray}
M_n^{(1)}(\alpha) &=& n^{-1} \sum_{i=1}^n w_{\alpha i}^c\rho\bigg\{
  \frac{y_i-h(x_{\alpha i}^T\hat{\beta}^c_{\alpha})}{\hat{\sigma}^{c} v(x_{\alpha_f i}^T\hat{\beta}^c_{\alpha_f})}\bigg\}  . \label{eqn: huberized in-sample criterion}
\end{eqnarray}   
Next, we implement a proportionally allocated, stratified  $m$-out-of-$n$ bootstrap of rows of $(y, X)$ in which we (i) compute and order the Pearson residuals, (ii) set the number of strata $K$ at between $3$ and $8$ depending on the sample size $n$, (iii) set stratum boundaries at the $$K^{-1}, 2K^{-1},\ldots, (K-1)K^{-1}$$ quantiles of the Pearson residuals, (iv) allocate observations to the strata in which the Pearson residuals lie, (v) sample $ \mbox{\#(observations in stratum } k)m/n$ (rounded as necessary) rows of $(y, X)$ independently with replacement from stratum $k$ so that the total sample size is $m$, (vi) use these data to construct the estimator $\hat{\beta}_{\alpha,m}^{c*}$, repeat steps (v) and (vi) $B$ independent times and then estimate the conditional expected prediction loss by $\hat{\sigma}^{c2}M_n^{(2)}(\alpha)$, where
\begin{eqnarray}
M_n^{(2)} (\alpha) &=& n^{-1} \E_* \sum_{i=1}^n w_{\alpha i}^c\rho\bigg(
  \frac{y_i-h[x_{\alpha i}^T\{\hat{\beta}^{c*}_{\alpha,m} - \E_*(\hat{\beta}^{c*}_{\alpha,m} - \hat{\beta}_{\alpha}^c)\}]}{\hat{\sigma}^{c} v(x_{\alpha_f i}^T\hat{\beta}^c_{\alpha_f})}\bigg), \label{eqn: huberized 1st est of efficient model}
  \end{eqnarray}  
where $\E_*$ denotes expectation with respect to the bootstrap distribution.  
In practice, it seems useful to take $m$ to be between $25-50\%$ of the sample size $n$ if working with moderate sample sizes, e.g. $50 \leq n \leq 200$. If $n$ is small then $m$ is small and the parameter estimators in the bootstrap do not converge for some bootstrap samples though this typically occurs less often with the stratified bootstrap. If $n$ is large then $m$ can be smaller than $25\%$ of the sample size $n$.
Combining (\ref{eqn: huberized in-sample criterion}) and (\ref{eqn: huberized 1st est of efficient model}), we estimate the criterion function (\ref{eqn:criterion}) by 
\begin{eqnarray}
M_n(\alpha) &=& \hat{\sigma}^{c2}\{M_n^{(1)}(\alpha)  +  n^{-1}\delta(n) p_{\alpha} + M_{n}^{(2)}(\alpha)\}.  \nonumber
\end{eqnarray}

The use of the stratified bootstrap ensures that we obtain bootstrap samples which are similar to the sample data in the sense that observations in the tails of the residual distribution and outliers are represented in each bootstrap sample or, with categorical data,  each category is represented in the bootstrap samples.  In essence, we construct an estimate of the conditional expected prediction loss based on samples which are similar to the sample we have observed.  The estimated variance function is estimated from a ``full'' model so does not change with the model $\alpha$.  This simplifies the procedure and has the advantage of making the procedure more stable.   Finally, we use the bias-adjusted bootstrap estimator $\hat{\beta}^{c*}_{\alpha,m} - \E_*(\hat{\beta}^{c*}_{\alpha,m} - \hat{\beta}_{\alpha}^c)$ rather than the bootstrap estimator $\hat{\beta}^{c*}_{\alpha,m}$ in $M_n^{(2)}(\alpha)$.  As discussed in more detail below, this achieves the same purpose as but avoids the centering technique used in M\"uller and Welsh (2005) and means that we do not have to include an intercept in every model.   It is therefore a useful refinement of the criterion given in M\"uller and Welsh (2005).

The computational burden of model selection can be reduced by limiting the number of different estimators we consider, reducing their computation by, for example, using good starting values from the initial fit to the data, and by reducing the number of models in ${\cal A}$.    Generally, our approach is to use an eclectic mix of methods including robust versions of deviance-tests, search schemes, diagnostics etc to produce a relatively small set ${\cal A}$ of competing models which we then compare using the methodology presented in this paper. In particular, we present a backward model search algorithm in Section 3.4 that substantially reduces the number of models to be considered while maintaining the consistency of $M_n(\alpha)$.

\section{Theoretical results}

Our procedure is intended to identify useful models whether or not a true model exists and our interest is not restricted to a single best model but to the identification of useful models (which make $M_n(\alpha)$ small).  In this context, if (i) a true model $\alpha_0$ exists and (ii) $\alpha_0 \subseteq \{1,\ldots,p\}$, then consistency in the sense that a procedure identifies $\alpha_0$ with probability tending to one is a desirable property.   Although in practice, we are interested in all the models which make $M_n(\alpha)$ small, we focus in this section on the model which minimises $M_n(\alpha)$ and show that choosing this model is consistent.  Specifically, for $c \in \mathcal{C}$, we define
\begin{equation}
\hat{\alpha}^{c}_{m,n} = \argmin_{\alpha\in\mathcal{A}} M_n(\alpha), \label{eqn: criterion}
\end{equation}
and develop conditions under which for each $c \in \mathcal{C}$,
\begin{equation}
\lim_{n\rightarrow\infty} \P\{ \hat{\alpha}^{c}_{m,n}  = \alpha_0 \}=1. \label{eqn:consistency}
\end{equation}

As in M\"uller and Welsh (2005), we define the subset of correct models $\mathcal{A}_c$ in $\mathcal{A}$ to be the set of models $\alpha \in \mathcal{A}$ such that $\alpha_0 \subseteq \alpha$; all other models are called incorrect models.  For any correct model $\alpha \in \mathcal{A}_c$, the errors $\epsilon_{\alpha i}=y_{i}-h(x_{\alpha i}^{T}\beta_{\alpha})$
 satisfy $\epsilon_{\alpha i} = \epsilon_{\alpha_0 i}$, for $i=1,\cdots,n$, and the components of $\beta_{\alpha}$ corresponding to columns of $X_{\alpha}$ which are not also in $\alpha_0$ equal zero.  

\subsection{Conditions}

It is convenient to introduce a simplified notation for stating the conditions and simplifying the proof of the  main result.  Write 
$$
h_{\alpha i} = h(x_{\alpha i} ^T \beta_{\alpha}),\,\,\,\,\, h'_{\alpha i} = h'(x_{\alpha i} ^T \beta_{\alpha}),\,\,\,\,\, h^{''}_{\alpha i} = h''(x_{\alpha i} ^T \beta_{\alpha}),
$$
$$
\sigma_i = \sigma v(x_{\alpha_f i} ^T \beta_{\alpha_f}),\,\,\,\,\, \epsilon_{\alpha_0 i} = \epsilon_{i},
$$ 
$$
\psi_i = \psi(\epsilon_i/\sigma_i),\,\,\,\mbox{ and }\,\,\, \psi'_i = \psi'(\epsilon_i/\sigma_i).
$$
Then we require the following conditions.

\begin{enumerate}

\item[(i)]  The $p_{\alpha} \times p_{\alpha}$ matrix 
$$\frac{1}{2n}\sum_{i=1}^n \sigma_{i}^{-2}w_{\alpha i}^c(h^{'2}_{\alpha i}\E\psi'_i- h^{''}_{\alpha i}\E\psi_i) x_{\alpha i}x_{\alpha i}^T \rightarrow \Gamma_{\alpha}^c,$$ 
where $\Gamma_{\alpha}^c$ is of full rank.

%

\item[(ii)] For all models $\alpha \in {\cal A}$ (including the full model), the estimators $\hat{\beta}^c_{\alpha} - \beta_{\alpha} = O_p(n^{-1/2})$,  $\hat{\sigma}^c - \sigma = O_p(n^{-1/2})$ with $\sigma > 0$.  For all correct models $\alpha \in {\cal A}_c$,
\[
n\Var(\hat{\beta}_{\alpha}^c) = \Sigma_{\alpha}+o_p(1),
\]
where $\Sigma_{\alpha}$ is of full rank.

\item[(iii)] For all models $\alpha \in {\cal A}$, the bootstrap estimator $\hat{\beta}^{c*}_{\alpha m} \rightarrow \beta_{\alpha}$ in probability. 
For all correct models $\alpha \in {\cal A}_c$, 
\begin{eqnarray*}
m \Var_*(\hat{\beta}^{c*}_{\alpha m})= n\kappa^c \Var(\hat{\beta}^c_{\alpha}) + o_p(1) 
\end{eqnarray*}
and for any two correct models $\alpha_{1},\alpha_{2} \in \mathcal{A}_{c}$ such that $\alpha_{1}\subset\alpha_{2}$
\begin{equation}
\trace(\Sigma_{\alpha_{2}}\Gamma_{\alpha_{2}}) - \trace(\Sigma_{\alpha_{1}}\Gamma_{\alpha_{1}}) >0. \label{eqn: cond (ii)}
\end{equation}

\item[(iv)] The sequence $\delta(n) = o(n/m)$ and $m = o(n)$.

\item[(v)] The derivatives $\psi=\rho'$ and $\psi'$ exist, are uniformly continuous, bounded, $\Var(\epsilon_{i}\psi_i) < \infty$, and $\E\psi'(\epsilon_{i}) > 0$, $i=1,\ldots,n$.

\item[(vi)] The weights are bounded, $h$ and its first two derivatives are continuous, $\sigma$ and $v$ are both positive,  and $v'$ is bounded.

\item[(vii)] The $x_i$ are bounded.

\item[(viii)] For any incorrect model $\alpha$, 
\[
\liminf_{n\rightarrow\infty} M^{(1)}_{n}(\alpha) > \lim_{n\rightarrow\infty}M^{(1)}_{n}(\alpha_{0}) \quad\text{a.s.}
\]

\end{enumerate}

\smallskip\noindent
Condition (i) is a generalization of a standard condition for fitting regression models which we require for generalized linear models.  
Condition (ii) is satisfied by many estimators;  condition (\ref{eqn: cond (ii)}) restricts the estimators we can consider in $\mathcal{C}$ but allows us to include maximum likelihood and other estimators such as the Cantoni and Ronchetti (2001) estimator.  We refer to (\ref{eqn: cond (ii)}) as the monotonicity condition.   
Condition (iii) specifies the required properties of the bootstrap parameter estimator.   In contrast to M\"uller and Welsh (2005), we have adjusted the bootstrap estimator so we do not have to impose conditions on the asymptotic bias of the bootstrap estimator.  Combining conditions (ii) and (iii), we obtain $\Var_*(\hat{\beta}^{c*}_{\alpha m}) =m^{-1}\kappa^c\Sigma_{\alpha} + o_p(m^{-1})$. Conditions (v)-(vii) enables us to make various two-term Taylor expansions and to control the remainder terms. We require a higher level of smoothness than exhibited by the $\rho$-function (\ref{eqn:rho}) but there are many functions satisfying these properties.    Condition (viii) is a generalisation of Condition (C4) of Shao (1996) to allow a more general choice of $\rho(\cdot)$.

We have specified a simple set of sufficient conditions (particularly in conditions (v)-(vii)) which are appropriate for a robust $\rho$ function and generalized linear models.  However, we note that we can specify alternative conditions and simpler conditions for particular cases.  For example, we obtain alternative conditions if we allow the $x_i$ to be stochastic; see for example Shao (1996, Condition C3. b.).  We can simplify our conditions if we use the nonrobust function $\rho(x)=x^2$; again see Shao (1996, p661).  Even in the robust case, simpler conditions can be given for homoscedastic linear models because $h(x) = x$, $v(x)=1$.  These possibilities are somewhat tangential to our main purpose so we will not pursue them in this paper.

\begin{thm}\label{thm: ms for robust situations}
Under the above conditions, the consistency result (\ref{eqn:consistency}) holds.
\end{thm}

\begin{proof}[Proof of Theorem \ref{thm: ms for robust situations}]\hfill

The proof of this result is similar to that given in M\"uller and Welsh (2005).  The main term we need to deal with is the bootstrap term 
\begin{eqnarray*}
M_{n}^{(2)}(\alpha) &=& \frac{1}{n} \E_* \sum_{i=1}^n w_{\alpha i} \rho\bigg\{
 (y_i-h[x_{\alpha i}^T\{\hat{\beta}_{\alpha,m}^* - \E_*(\hat{\beta}_{\alpha,m}^* - \hat{\beta}_{\alpha})\}])/\hat{\sigma}_i\bigg\},
\end{eqnarray*} 
where $\hat{\beta}_{\alpha}$ and $\hat{\sigma}_i = \hat{\sigma} v(x_{\alpha_f i}^T\hat{\beta}_{\alpha_f})$ are constant with respect to the bootstrap.  We make a Taylor expansion of $\rho$ as a function of $\hat{\beta}_{\alpha,m}^* - \E_*(\hat{\beta}_{\alpha,m}^* - \hat{\beta}_{\alpha})$ about $\hat{\beta}_{\alpha}$, to obtain
\begin{eqnarray*}
M_{n}^{(2)}(\alpha) &=& \frac{1}{n}\sum_{i=1}^n w_{\alpha i}\rho[\{y_i-h(x_{\alpha i}^T\hat{\beta}_{\alpha})\}/\hat{\sigma}_{i}] \\
&& +  \E_* \frac{1}{2n}\sum_{i=1}^n \hat{\sigma}^{-2}_{i}w_{\alpha i}x_{\alpha i}^T (\hat{\beta}^*_{\alpha, m}  - \E_*\hat{\beta}_{\alpha,m}^* ) (\hat{\beta}^*_{\alpha, m} - \E_*\hat{\beta}_{\alpha,m}^* )^T x_{\alpha i}\\
&& \times \left( h'(x_{\alpha i}^T\bar{\beta}_{\alpha})^2   \psi'[\{ y_i - h(x_{\alpha i}^T\bar{\beta}_{\alpha})\}/\hat{\sigma}_{i}] - h''(x_{\alpha i}^T\bar{\beta}_{\alpha}) \psi[\{y_i - h(x_{\alpha i}^T\bar{\beta}_{\alpha})\}/\hat{\sigma}_{i}] \right)\\
&=& T_{1} + T_{2},  
\end{eqnarray*} 
where $|\bar{\beta}_{\alpha} - \hat{\beta}_{\alpha}| \leq |\hat{\beta}^{*}_{\alpha,m} -\hat{\beta}_{\alpha}|$.   This equation is analogous to (9) in M\"uller and Welsh (2005) except that we have eliminated the linear term by using the bias-adjusted bootstrap estimator.  We consider $T_1$ and $T_2$ in turn.

\noindent{\bf Order of $T_{2}$:}
Let
 \begin{eqnarray*}
\bar H_{\alpha i} &=& h'(x_{\alpha i}^T\bar{\beta}_{\alpha})^2   \psi'[\{y_i - h(x_{\alpha i}^T\bar{\beta}_{\alpha})\}/\hat{\sigma}_{i}] - h''(x_{\alpha i}^T\bar{\beta}_{\alpha}) \psi[\{y_i - h(x_{\alpha i}^T\bar{\beta}_{\alpha})\}/\hat{\sigma}_{i}] \\
&=& h'(x_{\alpha i}^T\bar{\beta}_{\alpha})^2   \psi'[\{\epsilon_i + h(x_{\alpha i}^T\beta_{\alpha}) - h(x_{\alpha i}^T\bar{\beta}_{\alpha})\}/\hat{\sigma}_{i}] - h''(x_{\alpha i}^T\bar{\beta}_{\alpha}) \psi[\{\epsilon_i + h(x_{\alpha i}^T\beta_{\alpha}) - h(x_{\alpha i}^T\bar{\beta}_{\alpha})\}/\hat{\sigma}_{i}]
\end{eqnarray*}
and write  
 \begin{eqnarray*}
T_2 &=&  \E_* \frac{1}{2n}\sum_{i=1}^n \hat{\sigma}^{-2}_{i}w_{\alpha i}x_{\alpha i}^T (\hat{\beta}^*_{\alpha, m} - \E_*\hat{\beta}^*_{\alpha, m}) (\hat{\beta}^*_{\alpha, m} - \E_*\hat{\beta}^*_{\alpha, m})^T x_{\alpha i}\bar H_{\alpha i}\\
&=& \frac{1}{2n}\sum_{i=1}^n \sigma^{-2}_{i}w_{\alpha i}x_{\alpha i}^T \Var_*(\hat{\beta}^*_{\alpha, m})x_{\alpha i} (h^{'2}_{\alpha i}  \E \psi'_i - h''_{\alpha i} \E\psi_i)\\
&& +\frac{1}{2n}\sum_{i=1}^n \sigma^{-2}_{i}w_{\alpha i}x_{\alpha i}^T \Var_*(\hat{\beta}^*_{\alpha, m}) x_{\alpha i} (h^{'2}_{\alpha i}   \psi'_i  - h''_{\alpha i} \psi_i  - h^{'2}_{\alpha i}  \E \psi'_i + h''_{\alpha i} \E\psi_i)\\
&& + \E_* \frac{1}{2n}\sum_{i=1}^n w_{\alpha i}x_{\alpha i}^T (\hat{\beta}^*_{\alpha, m} - \E_*\hat{\beta}^*_{\alpha, m}) (\hat{\beta}^*_{\alpha, m} - \E_*\hat{\beta}^*_{\alpha, m})^T x_{\alpha i} (\hat{\sigma}^{-2}_{i}\bar H_{\alpha i} - \sigma^{-2}_{i}h^{'2}_{\alpha i}   \psi'_i +  \sigma^{-2}_{i}h''_{\alpha i} \psi_i).
\end{eqnarray*}
Then 
\begin{eqnarray*}
\lefteqn{\frac{1}{2n}\sum_{i=1}^n \sigma^{-2}_{i}w_{\alpha i}x_{\alpha i}^T \Var_*(\hat{\beta}^*_{\alpha, m})x_{\alpha i} (h^{'2}_{\alpha i}  \E \psi'_i - h''_{\alpha i} \E\psi_i)}\\
&=& \frac{1}{2n} \trace\left\{\Var_* (\hat{\beta}^*_{\alpha, m})
    \sum_{i=1}^n \sigma^{-2}_{i} w_{\alpha i}x_{\alpha i}x_{\alpha i}^T (h^{'2}_{\alpha i}  \E \psi'_i- h''_{\alpha i} \E\psi_i)\right\} \\
&=& \frac{\kappa^c}{2m} \trace( \Sigma_{\alpha} \Gamma_{\alpha})+ o_{p}(m^{-1})
\end{eqnarray*}
by condition (iii) and the first part of condition (i).
Similarly
$$
\frac{1}{2n}\sum_{i=1}^n \sigma^{-2}_{i}w_{\alpha i}x_{\alpha i}^T \Var_*(\hat{\beta}^*_{\alpha, m}) x_{\alpha i} (h^{'2}_{\alpha i}   \psi'_i - h''_{\alpha i} \psi_i - h^{'2}_{\alpha i}  \E \psi'_i +  h''_{\alpha i} \E\psi_i) = o_p(m^{-1})
$$
by condition (iii) and the second part of condition (i).
Finally,
$$
| \E_* \frac{1}{2n}\sum_{i=1}^nw_{\alpha i}x_{\alpha i}^T (\hat{\beta}^*_{\alpha, m} - \E_*\hat{\beta}^*_{\alpha, m}) (\hat{\beta}^*_{\alpha, m} - \E_*\hat{\beta}^*_{\alpha, m})^T x_{\alpha i} ( \hat{\sigma}^{-2}_{i}\bar H_{\alpha i} -  \sigma^{-2}_{i}h^{'2}_{\alpha i}\psi'_i +   \sigma^{-2}_{i}h''_{\alpha i} \psi_i )| =  o_p(m^{-1})
$$
provided 
\begin{eqnarray*}
 \max_{1 \le i \le n}\sup_{\epsilon} \sup_{|t -\beta_{\alpha}| \le n^{-1/2}C}|\hat{\sigma}^{-2}_{i} h'(x_{\alpha i}^Tt)^2   \psi'[\{\epsilon + h(x_{\alpha i}^T\beta_{\alpha}) - h(x_{\alpha i}^T t)\}/\hat{\sigma}_{i}]-  \sigma^{-2}_{i} h^{'2}(x_{\alpha i}^T\beta_{\alpha})   \psi'(\epsilon /\sigma_{i})| = o_p(1)
\end{eqnarray*}
and
\begin{eqnarray*}
 \max_{1 \le i \le n}\sup_{\epsilon} \sup_{|t -\beta_{\alpha}| \le n^{-1/2}C}  |\hat{\sigma}^{-2}_{i} h''(x_{\alpha i}^Tt)   \psi[\{\epsilon + h(x_{\alpha i}^T\beta_{\alpha}) - h(x_{\alpha i}^T t)\}/\hat{\sigma}_{i}] -  \sigma^{-2}_{i} h''(x_{\alpha i}^T\beta_{\alpha})    \psi(\epsilon /\sigma_{i})| = o_p(1).
\end{eqnarray*}
Conditions (v)-(viii) ensure that these requirements hold.

\noindent{ \bf Order of $T_{1}$:}
Let $|\tilde{\beta}_{\alpha} - {\beta}_{\alpha}| \leq |\hat{\beta}_{\alpha} -{\beta}_{\alpha}|$, $|\tilde{\beta}_{\alpha_f} - {\beta}_{\alpha_f}| \leq |\hat{\beta}_{\alpha_f} -{\beta}_{\alpha_f}|$ and $|\tilde{\sigma} - \sigma| \le | \hat{\sigma} - \sigma|$.
Recall that $\sigma_i = \sigma v(h_{\alpha_f i})$ and write
$$
D(y_i, h_{\alpha i}, \sigma_{i})= \left(\begin{array}{c} - x_{\alpha i}h'_{\alpha i}\sigma_{i}^{-1}\psi\{(y_i-h_{\alpha i})/\sigma_{i}\} \\ -\sigma_{i}^{-2}v(h^{-1}(h_{\alpha_f i})) (y_i-h_{\alpha i})\psi\{(y_i-h_{\alpha i})/\sigma_{i}\} 
\\  -x_{\alpha_f i} \sigma_{i}^{-2}\sigma v'(h^{-1}(h_{\alpha_f i})) (y_i-h_{\alpha i})\psi\{(y_i-h_{\alpha i})/\sigma_{i}\} \end{array}\right).
$$
Then
\begin{eqnarray*}
\frac{1}{n}\sum_{i=1}^n w_{\alpha i}\rho\{(y_i-\hat h_{\alpha i})/\hat{\sigma}_{i}\} 
&=& \frac{1}{n}\sum_{i=1}^n w_{\alpha i}\rho(\epsilon_i/\sigma_i)+ \frac{1}{n}\sum_{i=1}^n w_{\alpha i}  (\hat{\beta}_{\alpha} - {\beta}_{\alpha}, \hat{\sigma} - \sigma, \hat{\beta}_{\alpha_f } - \beta_{\alpha_f})^T D(y_i, h_{\alpha i}, \sigma_{i})\\
& & + \frac{1}{n}\sum_{i=1}^n  w_{\alpha i}  (\hat{\beta}_{\alpha} - {\beta}_{\alpha}, \hat{\sigma} - \sigma, \hat{\beta}_{\alpha_f } - \beta_{\alpha_f})^T\{D(y_i, \tilde h_{\alpha i}, \tilde{\sigma}_{i}) - D(y_i, h_{\alpha i}, \sigma_{i})\}\\  \nonumber \\
&=& \frac{1}{n}\sum_{i=1}^n w_{\alpha i}\rho(\epsilon_i/\sigma_i) +O_{p}(n^{-1/2})
\end{eqnarray*}
provided
$$
 \max_{1 \le i \le n}\sup_{\epsilon}\sup_{|t -\beta_{\alpha}| \le n^{-1/2}C}  |\hat{\sigma}^{-1}_{i} h'(x_{\alpha i}^T t)   \psi[\{\epsilon + h(x_{\alpha i}^T\beta_{\alpha}) - h(x_{\alpha i}^T t)\}/\hat{\sigma}_{i}]  - \sigma_{i}^{-1}h'_{\alpha i}\psi(\epsilon/\sigma_{i}) | = o_p(1)
$$
\begin{eqnarray*}
\lefteqn{ \max_{1 \le i \le n}\sup_{\epsilon} \sup_{|t -\beta_{\alpha}| \le n^{-1/2}C} |\hat{\sigma}^{-1}v(x_{\alpha_f i}^T\hat{\beta}_{\alpha_f})^{-2}v'(x_{\alpha_f i}^T\hat{\beta}_{\alpha_f})   \{\epsilon + h(x_{\alpha i}^T\beta_{\alpha}) - h(x_{\alpha i}^T t)\}}\\
&& \times \psi[\{\epsilon + h(x_{\alpha i}^T\beta_{\alpha}) - h(x_{\alpha i}^T t)\}/\hat{\sigma}_{i}] - 
\sigma^{-1}v(x_{\alpha_f i}^T\beta_{\alpha_f i})^{-2} v'(x_{\alpha_f i}^T\beta_{\alpha_f})  \epsilon \psi(\epsilon/\sigma_{i})| = o_p(1)
\end{eqnarray*}
\begin{eqnarray*}
\lefteqn{\max_{1 \le i \le n}\sup_{\epsilon}\sup_{|t -\beta_{\alpha}| \le n^{-1/2}C}  |\hat{\sigma}^{-2} v(x_{\alpha_f i}^T\hat{\beta}_{\alpha_f})^{-1}  \{\epsilon + h(x_{\alpha i}^T\beta_{\alpha}) - h(x_{\alpha i}^T t\} \psi[\{\epsilon + h(x_{\alpha i}^T\beta_{\alpha})}\\
&&  - h(x_{\alpha i}^T t)\}/\hat{\sigma}_{i}] -  \sigma^{-2}v(x_{\alpha_f i}^T\beta_{\alpha_f})^{-1} \epsilon \psi(\epsilon/\sigma_{i})| = o_p(1).
\end{eqnarray*}
As for $T_2$, these results follow from conditions (v)-(vii).

Putting both terms together, it follows that
\begin{eqnarray}
M_{n}^{(2)}(\alpha)
&=&  \frac{1}{n}\sum_{i=1}^n w_{\alpha i}\rho\{(y_i-h(x_{\alpha i}^T\hat{\beta}_{\alpha}))/\hat{\sigma}_{i}\}  + \frac{\kappa^c}{2m} \trace( \Sigma_{\alpha} \Gamma_{\alpha}) + o_{p}(m^{-1}) \label{eqn asymptotics in M2}
\end{eqnarray}
and the proof is completed as in M\"uller and Welsh (2005).
\end{proof}

\subsection{The elimination of bias}\label{sec: bias}

One of the main difficulties in constructing model selection criteria like $M_n(\alpha)$ is removing the bias (equivalently the linear term) in the expansion of  $M_n^{(2)}(\alpha)$.  
Suppose that instead of the bias-adjusted bootstrap estimator $\hat{\beta}^{c*}_{\alpha,m} - \E_*(\hat{\beta}^{c*}_{\alpha,m} - \hat{\beta}_{\alpha}^c)$, we use the bootstrap estimator $\hat{\beta}^{c*}_{\alpha,m}$ in $M_n^{(2)}(\alpha)$.  Then when we expand $M_n^{(2)}(\alpha)$ as in Shao (1996), M\"uller and Welsh (2005) or the proof of Theorem 3.1, we obtain the linear term
\begin{equation}
\E_*(\hat{\beta}^*_{\alpha, m} - \hat{\beta}_{\alpha})^T  \frac{1}{n}\sum_{i=1}^n \hat{\sigma}_{i}^{-1}w_{\alpha i}x_{\alpha i}  h'(x_{\alpha i}^T\hat{\beta}_{\alpha}) \psi\{(y_i - h(x_{\alpha i}^T\hat{\beta}_{\alpha}))/\hat{\sigma}_{i}\}.  \label{eqn:linear term}
\end{equation}
As shown in M\"uller and Welsh (2005), the bias term $\E_*(\hat{\beta}^*_{\alpha, m} - \hat{\beta}_{\alpha})$ is typically a function of $\alpha$ with leading term $O_p(m^{-1})$, the same as the quadratic term in the expansion.  Since the quadratic term governs the selection of correct models, it is crucial that the linear term be at least of smaller order.

There are various ways to make (\ref{eqn:linear term}) of order $o_p(m^{-1})$.  Notice that ordinarily the mean in (\ref{eqn:linear term}) is asymptotic to 
$$
 \frac{1}{n}\sum_{i=1}^n \sigma_{i}^{-1} w_{\alpha i}x_{\alpha i}  h'_{\alpha i} \E \psi_{i}
$$ 
which is $O(1)$.  However, if $ \E \psi_{i} = 0$, then it can be $O_p(n^{-1/2})$ which can be made $o_p(m^{-1})$.  This is the approach used in Shao (1996).  It holds when $\psi(x)=x$ but this is a nonrobust choice and hence unappealing in general.  M\"uller and Welsh (2005) took a different approach in which they insisted that each model contain an intercept and then centered the explanatory variables so that they have mean zero and the bias is forced into the intercept.  In fact, the intercept can be eliminated by replacing the intercept of the bootstrap estimator by that of the estimator $\hat{\beta}_{\alpha}$ or by fixing the intercepts at the value of the intercept estimated under a ``full'' model.  This approach is much less attractive in the present more general context because the centering vector has to include estimates of $\sigma_{i}$ and $ \E \psi_{i}$  (which previously did not depend on $i$) and $h'_{\alpha i}$ (which was previously not present).  This means that the centering vector is stochastic and the centered explanatory variables cannot simply be conditioned on.  Even if we overcome these difficulties, we have to ensure that the criterion is consistent and the arguments given in the next subsection do not apply unless the model is fitted with the same covariates as the model selection criterion uses.  This approach is not therefore very attractive.  

A different approach would be to require as in M\"uller and Welsh (2005) that $\E_*(\hat{\beta}^*_{\alpha, m} - \hat{\beta}_{\alpha}) = m^{-1}B_{\alpha} +o_p(m^{-1})$, estimate $B_{\alpha}$ and then adjust the criterion by subtracting off an estimate of  (\ref{eqn:linear term}).  Although this will remove the bias, it will add a contribution to the quadratic term which will affect the arguments in the next subsection.  Also, it changes the criterion which then loses its natural interpretability.  It is far better to think in terms of adjusting the bootstrap estimator $\hat{\beta}^*_{\alpha, m}$ for bias.  We could do this by focussing on $B_{\alpha}$ (as we only need the leading term) but then we would need to derive and estimate $B_{\alpha}$ for each estimator we consider.  Fortunately, we have available the bias itself in the very natural form $\E_*(\hat{\beta}^*_{\alpha, m} - \hat{\beta}_{\alpha})$ and so we can remove the bias entirely without having to assume any particular form.  This is the solution that we have adopted in using the bias adjusted bootstrap estimator in $M_n^{(2)}(\alpha)$.

\subsection{The monotonicity of $ \trace( \Sigma_{\alpha} \Gamma_{\alpha})$}\label{sec: examples and counterexamples}

The assumption (iii) that $ \trace( \Sigma_{\alpha} \Gamma_{\alpha})$ is monotone in $p_{\alpha}$ does not hold in general for arbitrary positive semi--definit matrices $\Sigma_{\alpha}$ and $\Gamma_{\alpha}$.   For example, with
\begin{equation*}
\Sigma_{\alpha_{2}} = \begin{pmatrix} 1.0 & -0.5 \\ -0.5 & 1.0 \end{pmatrix}, \quad
\Gamma_{\alpha_{2}} = \begin{pmatrix} 1.0 & 0.2 \\ 0.2 & 0.1 \end{pmatrix}, \quad
\Sigma_{\alpha_{1}} = 1.0, \quad
\Gamma_{\alpha_{1}} = 1.0,
\end{equation*}
and we find that
\begin{equation*}
\trace(\Sigma_{\alpha_{2}}\Gamma_{\alpha_{2}})
- \trace(\Sigma_{\alpha_{1}}\Gamma_{\alpha_{1}}) = 0.9 - 1.0 = -0.1.
\end{equation*}
However, M\"uller and Welsh (2005) prove that for linear regression models, the condition holds for the class of Mallows type M--estimators or one--step Mallows type M--estimators etc., because of the relationship between $\Var(\hat{\beta}_{\alpha})$ and $\Gamma_{\alpha}$.

\medskip
Consider the maximum likelihood estimator for generalized linear models.  We can write condition (i) as
$$
n^{-1} X_{\alpha}^{T}W_{\Gamma_{\alpha}}X_{\alpha}\rightarrow \Gamma_{\alpha},
$$
where $W_{\Gamma_{\alpha}} = \frac{1}{2}\operatorname{diag}(\sigma_1^{-2}w_{\alpha 1}(h^{'2}_{\alpha 1}\E\psi'_{1}- h^{''}_{\alpha 1}\E\psi_{1}),\ldots,\sigma_n^{-2}w_{\alpha n}(h^{'2}_{\alpha n}\E\psi'_n - h^{''}_{\alpha n}\E\psi_n))$.  From McCullagh and Nelder (1989, p43), the maximum likelihood estimator $\hat{\beta}_{\alpha}$ satisfies
\begin{equation*}
n\Var(\hat{\beta}_{\alpha}) = (X_{\alpha}^{T}W_{\Sigma_{\alpha}}X_{\alpha})^{-1} + o_{p}(1),
\end{equation*}
where $W_{\Sigma_{\alpha}} = \operatorname{diag}(h^{'2}_{\alpha 1}/\sigma_1^2 ,\ldots,h^{'2}_{\alpha n}/\sigma_n^2)$. We have to show that 
$$
\trace\Big\{(X_{\alpha}^{T}W_{\Sigma_{\alpha}} X_{\alpha})^{-1} X_{\alpha}^{T}W_{\Gamma_{\alpha}}X_{\alpha}\Big\}
$$ 
is strictly monotone increasing in $p_{\alpha}$. 

Reorder the rows of $X_{\alpha}$ if necessary so that the top $p_{\alpha} \times p_{\alpha}$ submatrix $C_{\alpha}$ is nonsingular.  Then the $p_{\alpha} \times n$ matrix $X_{\alpha}^{-} = (C_{\alpha}^{-1}, 0)$ is a generalized  inverse of $X_{\alpha}$. Then we have that
\begin{eqnarray*}
\trace\Big( X_{\alpha}^{-} W_{\Sigma_{\alpha}}^{-1} X_{\alpha}^{-T}  X_{\alpha}^{T}W_{\Gamma_{\alpha}}X_{\alpha}\Big) 
&=& \trace\Big(X_{\alpha} X_{\alpha}^{-} W_{\Sigma_{\alpha}}^{-1} X_{\alpha}^{-T}  X_{\alpha}^{T}W_{\Gamma_{\alpha}}\Big) \\
 &=& \trace\Big(X_{\alpha} X_{\alpha}^{-} W_{\Sigma_{\alpha}}^{-1}  X_{\alpha} X_{\alpha}^{-} W_{\Gamma_{\alpha}}\Big).
\end{eqnarray*}
By definition of the generalized  inverse, $X_{\alpha} X_{\alpha}^{-}$ is a symmetric $n\times n$ matrix with first $p_{\alpha}$ diagonal elements equal to $+1$ and the remaining elements zero so that $X_{\alpha} X_{\alpha}^{-} = \operatorname{diag}(1,\ldots,1,0,\ldots,0)$. Therefore,
\begin{eqnarray*}
\trace\Big\{(X_{\alpha}^{T}W_{\Sigma_{\alpha}} X_{\alpha})^{-1} X_{\alpha}^{T}W_{\Gamma_{\alpha}}X_{\alpha}\Big\} &=& \frac{1}{2} \sum_{i=1}^{p_{\alpha}} w_{\alpha i} \frac{h^{'2}_{\alpha i}\E\psi'_i - h^{''}_{\alpha i}\E\psi_i}{h^{'2}_{i}} .
\end{eqnarray*}
The simplest sufficient condition for monotonicity is 
\begin{eqnarray}
h^{'2}_{\alpha i}\E\psi'_i > h^{''}_{\alpha i}\E\psi_i, \quad  i= 1,\dots, n.  \label{eqn:monotonicity}
\end{eqnarray}
Since the left hand side is positive, it suffices to show that $h^{''}_{\alpha i}\E\psi_i \le 0$ for $i= 1,\dots, n$.

The monotonicity condition (\ref{eqn:monotonicity}) holds if $\E\psi_i = 0$ or $h^{''}_{\alpha i} = 0$.  The first case occurs when (i) $\rho(x) = x^2$  or (ii) the $\epsilon_i = y_i - h_{\alpha i}$ has a distribution which is symmetric about zero and $\psi$ is antisymmetric and the second when we use the identity link so $h(x) =x$.   Shao (1996) exploited (i) but this choice favours least squares estimation and is non-robust so we prefer not to use it; (ii) applies to Gaussian models but not to models with asymmetric distributions.  Similarly, the identity link is widely used in Gaussian models and may be used in gamma models but is not useful in binomial and Poisson models.  In these cases, we need to examine (\ref{eqn:monotonicity}) more carefully.  For the log link which is often used in Poisson and gamma models
$$
h(\eta_{\alpha i}) =  h'(\eta_{\alpha i}) = h''(\eta_{\alpha i}) =  \exp(\eta_{\alpha i}) > 0
$$
and for the reciprocal link which is often used in gamma models
$$
h(\eta_{\alpha i}) =  \frac{1}{ \eta_{\alpha i}},\quad h'(\eta_{\alpha i}) =  - \frac{1}{\eta_{\alpha i}^{2}},\quad h''(\eta_{\alpha i}) = \frac{2}{\eta_{\alpha i}^{3}} > 0.
$$
However, for many right skewed distributions like the Poisson and gamma, anti-symmetric $\psi$ functions with sufficiently large $b$ truncate more of the upper tail than the lower tail so  $\E\psi_i \le 0$. To see this, note that for the $\psi$ function (\ref{eqn:rho}), we can write
\begin{eqnarray*}
\E\psi_i &=&\int_{-\mu/\sigma}^{\infty} \psi(z) dF(\sigma z + \mu)\\
&=&  \int_{-\min(b,\mu/\sigma)}^b 2z  dF(\sigma z + \mu) \\
&=&  -\int_{-\mu/\sigma}^{-\min(b,\mu/\sigma)} 2z dF(\sigma z + \mu) - \int_{b}^{\infty} 2z dF(\sigma z + \mu)\\
&\le & 0,
\end{eqnarray*}
provided $b$ is large enough to ensure that $\int_{-\mu/\sigma}^{-\min(b,\mu/\sigma)} z dF(\sigma z + \mu) + \int_{b}^{\infty} z dF(\sigma z + \mu) \ge 0$.
It follows that $h^{''}_{\alpha i}\E\psi_i \le 0$  and (\ref{eqn:monotonicity}) holds in these cases. 
For the logistic link
$$
h(\eta_{\alpha i}) = \frac{\exp(\eta_{\alpha i})}{ 1 + \exp(\eta_{\alpha i})},\quad h'(\eta_{\alpha i}) = \frac{\exp(\eta_{\alpha i})}{(1+\exp(\eta_{\alpha i}))^{2}},\quad h''(\eta_{\alpha i}) = \frac{\exp(\eta_{\alpha i})-\exp(2\eta_{\alpha i})}{(1+\exp(\eta_{\alpha i}))^{3}}
$$
so that $h_{\alpha i} <1/2$, $h^{''}_{\alpha i} > 0$ if $\eta_{\alpha_i} < 0$ and $h_{\alpha i} > 1/2$, $h^{''}_{\alpha i} < 0$ if $\eta_{\alpha_i} > 0$ and we need a more careful analysis.  The Bernoulli model can be left or right skewed depending on the value of $h_{\alpha i}$ so $\E\psi_i $ can be positive or negative.  Fortunately, for anti-symmetric $\psi$,
\begin{eqnarray*}
\E\psi_i &=& \E \psi\{(y_i - h_{\alpha_i})/h'_{\alpha_i}\} \\
&=& \psi( - h_{\alpha_i}/h'_{\alpha_i})(1- h_{\alpha_i}) + \psi\{(1- h_{\alpha_i})/h'_{\alpha_i}\}  h_{\alpha_i}\\
&=& -  \psi(h_{\alpha_i}/h'_{\alpha_i})(1- h_{\alpha_i}) + \psi\{(1- h_{\alpha_i})/h'_{\alpha_i}\}  h_{\alpha_i}
\end{eqnarray*}
from which
$\E\psi_i  \le 0$ if $\eta_{\alpha_i} < 0$ and  $\E\psi_i  \ge  0$ if $\eta_{\alpha_i} > 0$ so that $h^{''}_{\alpha i}\E\psi_i \le 0$  and (\ref{eqn:monotonicity}) holds.


Next, we consider the quasi--likelihood estimator for the logistic model as defined in Cantoni and Ronchetti (2001, Section 2.2).  The Mallows quasi--likelihood estimator is the solution of the estimating equations,
\begin{equation}
\sum_{i=1}^{n} \bigg[ w(x_{\alpha i}) x_{\alpha i}\frac{1}{v(x_{\alpha i}^T\beta_{\alpha})} h'(x_{\alpha i}^T\beta_{\alpha})\psi_{c}(r_{\alpha i}) - a(\beta) \bigg] = 0,
\end{equation}
where $r_{i} = \{y_{i}-h(x_{\alpha i}^T\beta_{\alpha})\}/v(x_{\alpha i}^T\beta_{\alpha})$ are the Pearson residuals, $\psi_{c}$ is the Huber function defined by
\begin{equation*}
\psi_{c}(r) = \begin{cases} r,          & |r| \leq c,\\
                            c\operatorname{sign}(r),  & |r| >    c, \end{cases}
\end{equation*}
and
\begin{equation*}
a(\beta) = \frac{1}{n} \sum_{i=1}^{n}w(x_{\alpha i}) x_{\alpha i}\frac{1}{v(x_{\alpha i}^T\beta_{\alpha})} h'(x_{\alpha i}^T\beta_{\alpha})\E \psi_{c}(r_{\alpha i}) .
\end{equation*}
When $w(x_{i}) = 1$,  the estimator is called the Huber quasi--likelihood estimator.  In general we do not require that $\psi_{c} = \rho' = \psi$ or that $w_{\alpha i} = w(x_{\alpha i})$. 
Cantoni and Ronchetti (2001, Appendix B) show that the estimator has an asymptotic normal distribution with asymptotic variance $\Sigma_{\alpha} = M_{\alpha}^{-1} Q_{\alpha} M_{\alpha}^{-1}$, where
\begin{equation*}
Q_{\alpha} = \frac{1}{n} X_{\alpha}^{T} A X_{\alpha} - a(\beta) a(\beta)^{T} \quad\text{and}\quad 
M_{\alpha} = \frac{1}{n} X_{\alpha}^{T} B X_{\alpha},
\end{equation*}
with $A$ and $B$ are diagonal matrices with diagonal elements
\begin{eqnarray*}
a_{ii} &=&w(x_{\alpha i})^2 \frac{1}{\sigma^2v(x_{\alpha i}^T\beta_{\alpha})^2} h'(x_{\alpha i}^T\beta_{\alpha})^2\E\psi_{c}(r_{\alpha i})^2, \\
b_{ii} &=& w(x_{\alpha i}) \frac{1}{\sigma v(x_{\alpha i}^T\beta_{\alpha})} h'(x_{\alpha i}^T\beta_{\alpha})^2 \E r_{\alpha i}\psi_{c}(r_{\alpha i}).
\end{eqnarray*}
Using the same generalized inverse as before so that
$$X_{\alpha}X_{\alpha}^{-} = \operatorname{diag}(1,\ldots,1,0,\ldots,0) = E_{n,p_{\alpha}},$$
we have to show that
$\trace( M_{\alpha}^{-1} Q_{\alpha} M_{\alpha}^{-1} \Gamma_{\alpha} )$ is monotone in $p_{\alpha}$. 
Indeed,
\begin{eqnarray*}
\trace( M_{\alpha}^{-1} Q_{\alpha} M_{\alpha}^{-1} \Gamma_{\alpha} )&=& \trace( X^{-}_{\alpha}B^{-1}X_{\alpha}^{-T} X_{\alpha}^{T} A X_{\alpha} X^{-}_{\alpha}B^{-1}X_{\alpha}^{-T} X_{\alpha}^{T}W_{\Gamma_{\alpha}} X_{\alpha} ) \\
&=& \trace( X_{\alpha}X^{-}_{\alpha}B^{-1}X_{\alpha}^{-T} X_{\alpha}^{T} A X_{\alpha} X^{-}_{\alpha}B^{-1}X_{\alpha}^{-T} X_{\alpha}^{T}W_{\Gamma_{\alpha}} ) \\
&=& \trace( E_{n,p_{\alpha}} B^{-1} E_{n,p_{\alpha}} A E_{n,p_{\alpha}} B^{-1} E_{n,p_{\alpha}} W_{\Gamma_{\alpha}} )\\
&=& \frac{1}{2} \sum_{i=1}^{p_{\alpha}}  \frac{a_{ii}(h^{'2}_{\alpha i}\E\psi'_{i}- h^{''}_{\alpha i}\E\psi_i)}{b_{ii}^{2}} 
\end{eqnarray*}
is a monotone function in $p_{\alpha}$.  This function is monotone in $p_{\alpha}$ under the same conditions as the analogous function for maximum likelihood estimation.

\subsection{The reduction of models} 
For any incorrect model  $\alpha\in\mathcal{A}\setminus\mathcal{A}_c$ it follows from condition (vi) in Section 3.1 and from (\ref{eqn asymptotics in M2}) that
\begin{eqnarray}
\liminf_{n\rightarrow\infty} M_{n}(\alpha) > \lim_{n\rightarrow\infty}M_{n}(\alpha_{0}) \quad\text{a.s.} \nonumber
\end{eqnarray}
and for any correct model $\alpha\in\mathcal{A}_c$
\begin{eqnarray}
M_{n}(\alpha)
&=&  \frac{1}{n}\sum_{i=1}^n w_{\alpha i}\rho\{(y_i-h(x_{\alpha i}^T\hat{\beta}_{\alpha}))/\hat{\sigma}_{i}\}  + \frac{\kappa^c}{2m} \trace( \Sigma_{\alpha} \Gamma_{\alpha}) + o_{p}(m^{-1}) \nonumber
\end{eqnarray}
Hence, it follows that for fixed $p_{\alpha_f}$ we also have 
\begin{eqnarray}
\liminf_{n\rightarrow\infty} \min_{\alpha\in\mathcal{A}\setminus\mathcal{A}_c}M_{n}(\alpha) > \lim_{n\rightarrow\infty}\max_{\alpha\in\mathcal{A}_c}M_{n}(\alpha_{0}) \quad\text{a.s.} \label{eqn: model reduction}
\end{eqnarray}

Equation (\ref{eqn: model reduction}) ensures that backward model selection schemes based on $M_n(\alpha)$ maintain consistency for the true model if $\mathcal{A}$ is the set of all possible $2^{p_{\alpha_f}}$ submodels. In particular we suggest using the following backward selection algorithm if the number of submodels to be considered is too large to be dealt with in practical problems.

\medskip \noindent
{\bf Algorithm 3.1.}
\begin{enumerate}
\item Calculate $M_n(\alpha)$ for the full model $\alpha_f =\{1,\ldots,p_{\alpha_f}\}$ and $\alpha_{f,-i} =\{1,\ldots,p_{\alpha_f}\}\setminus\{i\}$, $i=1,\ldots,p_{\alpha_f}$, resulting in $\{M_n(\alpha): \# \alpha \geq p_{\alpha_f}-1\}$.
\item Set $\alpha_f = \argmin_{\{\# \alpha \geq p_{\alpha_f}-1\}} M_n(\alpha)$ and repeat 1.\ if $\alpha_f\geq 2$.
\item Estimate $\alpha$ by the $\argmin$ of $M_n(\alpha)$ over all $1+\sum_{i=1}^{p_{\alpha_f}}i=1+k(k+1)/2$ considered models.
\end{enumerate}

An example of the solution paths of all submodels and of the backward selected submodels is given in Figure \ref{fig1} in Section 5.

\section{Simulation study}

In this section we present a range of simulation results for Poisson regression models. The proposed robust model selection criterion based on robust and non robust parameter estimator procedures with $b=2$ and $\delta(n) = 2\log(n)$ is compared to the AIC and BIC criteria. 

We generated data according to the Poisson regression model
\begin{equation}
\eta_{i} = \log \mu_{i}
= \beta_{1} + \beta_{2}x_{2i} + \beta_{3}x_{3i}  + \beta_{4}x_{4i}, \quad y_{i}\sim Poi(\mu_{i}), \quad i=1,\ldots,n, \label{sim model 1}
\end{equation}
with true parameter vectors $(1,0,0,0)$, $(-1,2,0,0)$, and $(-1,1,1,0)$ such that $\sum_j \beta_j =1$. The response variable is Poisson distributed with mean $\mu_{i}$. The explanatory data is generated by drawing pseudo--random numbers from the multivariate normal with mean vector $(1,1,1)$ and covariance matrix given by diagonal elements equal to 1 and off diagonal elements equal to 0. 

In this non--robust setting we generated for each of the $500$ simulation runs $n=64$ data points and estimated the parameters by $\hat{\beta}_{ML}$ using the \texttt{glm.fit} (ML estimator) and by $\hat{\beta}_{CR}$ using the {\tt glmrob} (Mallows or Huber type robust estimators; see Cantoni and Ronchetti, 2001) function in R.  We calculated  AIC, BIC, and the proposed robust model selection criteria with 8 equally sized strata based on the Pearson residuals from the full model,

The bootstrap estimators for $m =24$ are based on $B=50$ bootstrap samples. Selection probabilities are presented in Table \ref{table 1}. Note that for $500$ simulations the empirical standard deviations for the empirical selection probabilities $\hat{\pi}$ are given by
$sd_{\hat{\pi}}=\sqrt{\hat{\pi}(1-\hat{\pi})/500} < 0.023.$

\begin{center}
\emph{Put Table \ref{table 1} around here.}
\end{center}

In this non--robust simulation the overall performance of the selection criteria $\hat{\alpha}$ is superior to classical criteria such as the AIC and BIC criterion independently of the chosen estimation procedure. As an example consider the results for the true parameter vector $(1,0,0,0)$ where the selection probabilities of the true model using $\hat{\beta}_{ML}$ are $0.58$ for AIC, $0.60$ for BIC, $0.90$ for $\hat{\alpha}$, and using $\hat{\beta}_{CR}$ the estimated probability is $0.89$ for $\hat{\alpha}$.

\bigskip
Next we generated data according to the model in equation (\ref{sim model 1}) but we added 8 moderate outliers in the response for the 8 observations with largest explanatory variable $x_4$. That is if  $\text{rank}(x_{4i}) := \sum_{k=1}^n \Ind(x_{4k}\leq x_{4i})\geq 57$ then $y_{i} \sim Poi(10)$, $i=1,\ldots,n$. All other simulation specifications remain the same. The selection probabilities are presented in Table \ref{table 2}.

\begin{center}
\emph{Put Table \ref{table 2} around here.}
\end{center}

In Table \ref{table 2} we see that the proposed selection criterion used with the robust estimator from Cantoni and Ronchetti (2001) performs outstandingly well. Used with the maximum likelihood estimator, it still performs very well compared to AIC and BIC. As an example consider the results for the true parameter vector $(-1,2,0,0)$. The selection probabilities of the true model using $\hat{\beta}_{ML}$  are  $0.01$ for AIC, 	$0.01$ for BIC, $0.66$ for $\hat{\alpha}$, and using $\hat{\beta}_{CR}$ the selection probability equals $0.78$ for $\hat{\alpha}$. 

\bigskip
Finally, we generated data according to the model in equation (\ref{sim model 1}) but we added 2 influential outliers in the response variable according to the condition $\text{rank}(x_{4i}) \leq 2$ then $y_{i} \sim Poi(100)$, $i=1,\ldots,n$. All other simulation specifications remain the same. The selection probabilities are presented in Table \ref{table 3}.
\begin{center}
\emph{Put Table \ref{table 3} around here.}
\end{center}
Table \ref{table 3} shows that the robust model selection criterion can break down if it is used with $\hat{\beta}_{ML}$ but still perform well with robust parameter estimators. As an example consider the results for the true parameter vector $(-1,1,1,0)$. The selection probabilities of the true model using $\hat{\beta}_{ML}$ equals 0 for AIC, BIC, and $\hat{\alpha}$. On the other hand, using $\hat{\beta}_{CR}$ the estimated probability is $0.71$ for $\hat{\alpha}$.

\section{Real data example}\label{sec real data examples}
In this section we present a real data example on the diversity of arboreal marsupials (possums) in the montane ash forest (Australia) which is part of the {\tt robustbase} package in R ({\tt possumDiv.rda}). The dataset is extensively described by Lindenmayer et al.\ (1990, 1991) and serves as a generalized linear model example with a canonical link function having Poisson distributed responses conditional on the linear predictor (Weisberg and Welsh, 1994; Cantoni and Ronchetti, 2001). The number of of different species ({\tt diversity}, count variable, mean $=1.48$, range $=  0-5$) was observed on $n=151$ sites. The explanatory variables describe the sites in terms of the number of shrubs ({\tt shrubs}, count variable, $5.06$, $0-21$), number of cut stumps from past logging operations ({\tt stumps}, count variable, $0.09$, $0-1$), the number of stags ({\tt stags}, count variable, $7.24$, $0-31$), a bark index ({\tt bark}, ordinal variable, $7.91$, $0-29$), the basal area of acacia species ({\tt acacia}, ordinal variable, $4.83$, $0-10$), a habitat score ({\tt habitat}, ordinal variable, $11.96$, $0-39$), the species of eucalypt with the greatest stand basal area ({\tt eucalypt}, nominal variable, three categories), and the aspect of the site ({\tt aspect}, nominal variable, four categories). We calculate $\hat{\alpha}$ based on $\hat{\beta}_{CR}$ with the same specifications as in the simulation study but because $n$ is considerably larger than $64$ we choose a smaller proportion for the bootstrap. That is $m=40$ which is about $26\%$ of the sample size. The best model according to our criterion $M_n(\alpha)$ includes {\tt stags} and {\tt habitat} which are also selected if the backward selection algorithm in Section 3.4 is applied. The solution paths of $M_n(\alpha)$ is given in Figure 1 which shows the minimal value of $M_n(\alpha)$ for all considered models with the same number of variables. 


\begin{figure}[!h]
\begin{center}
\centerline{\includegraphics[width=16cm]{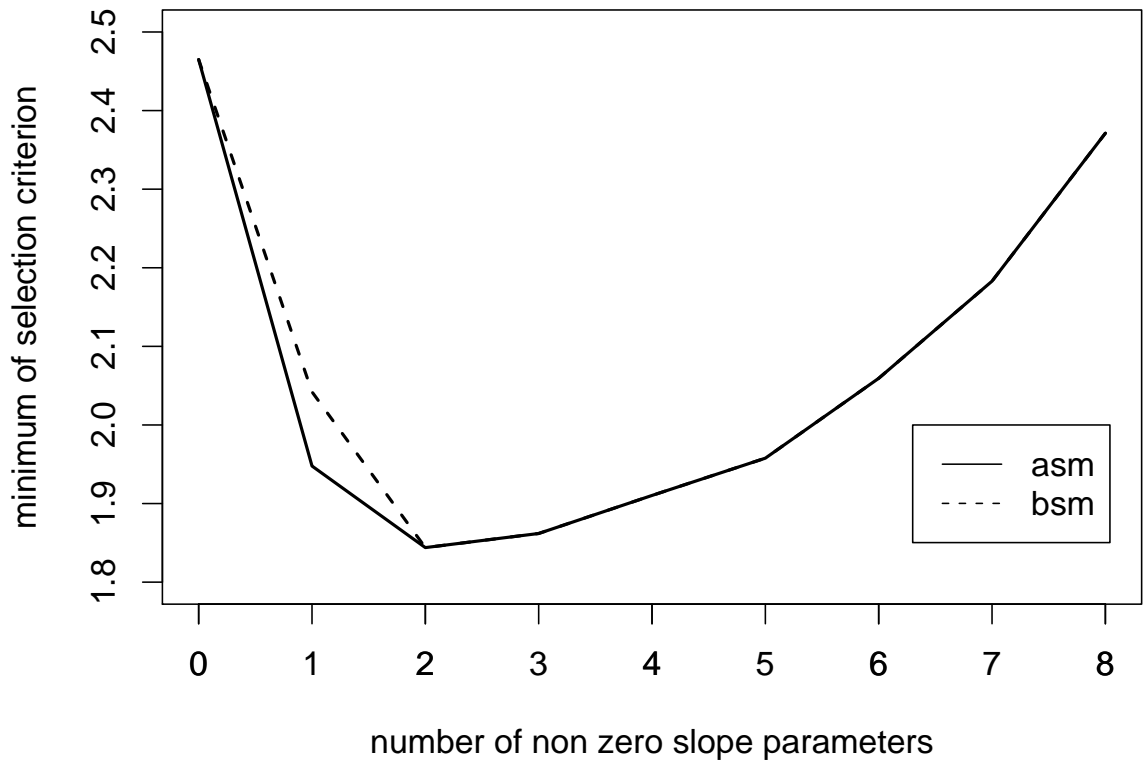}}
\end{center}
\vspace*{-1.5cm}
\hspace*{1.75cm}\begin{minipage}[t]{0.8\linewidth}
\caption{Solution path for the minimum of $M_n(\alpha)$ given a fixed number of non zero slope parameters for all submodels (asm) and for backward selected submodels (bsm).}\label{fig1}\end{minipage}
\end{figure}

Cantoni and Ronchetti (2001) mentioned that there are four potentially influential data points, namely, observations 59, 110, 133, and 139. According to the results of our simulation study we therefore consider $\hat{\alpha}$ together with $\hat{\beta}_{CR}$ to be superior to AIC, BIC, and  $\hat{\alpha}$ with $\hat{\beta}_{ML}$. 
Table \ref{table 4} presents an overview of the estimated best model which includes also the results of Cantoni and Ronchetti (2001, Section 5.2). 

\begin{center}
\emph{Put Table \ref{table 4} around here.}
\end{center}

\section{Discussion and conclusions}
We have proposed a bootstrap criterion for robustly selecting generalized linear models.  The criterion is a generalization of that developed for regression models by M\"uller and Welsh (2005) and has its strengths while still improving on that criterion.  In particular, the criterion (i) combines a robust penalised criterion (which reflects goodness-of-fit to the data) with an estimate of a robust measure of the conditional expected prediction error (which measures the ability to predict as yet unobserved observations), (ii) separates the comparison of models from any particular method of estimating them, and (iii) uses the stratified bootstrap to make the criterion more stable.   The improvement is achieved by using the bootstrap to estimate the bias of the bootstrap estimator of the regression parameter and then using the bias-adjusted bootstrap estimator instead of the raw bootstrap estimator in the criterion.  This step widens the applicability of the method by removing the requirement of M\"uller and Welsh (2005) that the models under consideration include an intercept.  We have also developed a more widely applicable method than that given in M\"uller and Welsh (2005) for establishing that the criterion can be applied with particular robust estimators of the regression parameters.  Our main theoretical result established the asymptotic consistency of the method and the simulation study shows that the model selection method works very well in finite samples.




%


\newpage
\begin{table}[ht]
\caption{Estimated selection probabilities based on the maximum--likelihood estimator $\hat{\beta}_{ML}$ and on the robust estimator $\hat{\beta}_{CR}$ from Cantoni and Ronchetti (2001). The results are based on $500$ Monte Carlo simulations and the bootstrap is based on $B=50$ replications. The data has no outlying points.} \label{table 1}
\begin{center}
\begin{tabular}{ccc|ccc|c}
\hline
& & &\multicolumn{3}{c|}{$\hat{\beta}_{ML}$} & $\hat{\beta}_{CR}$ \\
true $\beta^T$ & model & type & AIC & BIC &  $\hat{\alpha}^{s_{8}}_{m,n}$ &  $\hat{\alpha}^{s_{8}}_{m,n}$ \\
\hline
$(1,0,0,0)$ & $(\beta_{1},0,0,0)$                   			& $\alpha_{0}$         & 0.58	&	0.60	&	0.90	&	0.89  \\
                 & $(\beta_{1},\beta_{2},0,0)$                   	& $\mathcal{A}_c$	& 0.10	&	0.10	&	0.02	&	0.03  \\
                 & $(\beta_{1},0,\beta_{3},0)$                   	& $\mathcal{A}_c$	& 0.13	&	0.13	&	0.04	&	0.05 \\
                 & $(\beta_{1},0,0,\beta_{4})$                   	& $\mathcal{A}_c$	& 0.13	&	0.12	&	0.03	&	0.03 \\
                 & $(\beta_{1},\beta_{2},\beta_{3},0)$             	& $\mathcal{A}_c$	& 0.03	&	0.02	&	0.00	&	0.00  \\
                 & $(\beta_{1},\beta_{2},0,\beta_{4})$             	& $\mathcal{A}_c$	& 0.02	&	0.02	&	0.00	&	0.00 \\
                 & $(\beta_{1},0,\beta_{3},\beta_{4})$             	& $\mathcal{A}_c$	& 0.02	&	0.02	&	0.00	&	0.00  \\
                 & $(\beta_{1},\beta_{2},\beta_{3},\beta_{4})$  	& $\mathcal{A}_c$	& 0.00	&	0.00	&	0.00	&	0.00  \\
\hline
$(-1,2,0,0)$ & $(\beta_{1},0,0,0)$                   			& --                      	& 0.00	&	0.00	&	0.00	&	0.00  \\
                 & $(\beta_{1},\beta_{2},0,0)$                   	& $\alpha_{0}$         & 0.65	&	0.67	&	0.94	&	0.93  \\
                 & $(\beta_{1},0,\beta_{3},0)$                   	& --                      	& 0.00	&	0.00	&	0.00	&	0.00 \\
                 & $(\beta_{1},0,0,\beta_{4})$                   	& --                      	& 0.00	&	0.00	&	0.00	&	0.00 \\
                 & $(\beta_{1},\beta_{2},\beta_{3},0)$             	& $\mathcal{A}_c$	& 0.15	&	0.15	&	0.03	&	0.03  \\
                 & $(\beta_{1},\beta_{2},0,\beta_{4})$             	& $\mathcal{A}_c$	& 0.17	&	0.16	&	0.03	&	0.03 \\
                 & $(\beta_{1},0,\beta_{3},\beta_{4})$             	& --	      			& 0.00	&	0.00	&	0.00	&	0.00  \\
                 & $(\beta_{1},\beta_{2},\beta_{3},\beta_{4})$  	& $\mathcal{A}_c$	& 0.03	&	0.02	&	0.00	&	0.00  \\
\hline
$(-1,1,1,0)$ & $(\beta_{1},0,0,0)$                   		 	& --                      	& 0.00	&	0.00	&	0.00	&	0.00 \\
                 & $(\beta_{1},\beta_{2},0,0)$                   	& --                      	& 0.00	&	0.00	&	0.00	&	0.00  \\
                 & $(\beta_{1},0,\beta_{3},0)$                   	& --                      	& 0.00	&	0.00	&	0.05	&	0.07  \\
                 & $(\beta_{1},0,0,\beta_{4})$                   	& --                      	& 0.00	&	0.00	&	0.00	&	0.00  \\
                 & $(\beta_{1},\beta_{2},\beta_{3},0)$             	& $\alpha_{0}$      	& 0.81	&	0.82	&	0.91	&	0.89 \\
                 & $(\beta_{1},\beta_{2},0,\beta_{4})$             	& --                      	& 0.00	&	0.00	&	0.00	&	0.00 \\
                 & $(\beta_{1},0,\beta_{3},\beta_{4})$             	& --                      	& 0.00	&	0.00	&	0.00	&	0.00 \\
                 & $(\beta_{1},\beta_{2},\beta_{3},\beta_{4})$  	& $\mathcal{A}_c$	& 0.19	&	0.18	&	0.03	&	0.04  \\
\end{tabular}
\end{center}
\end{table}

\newpage
\begin{table}[ht]
\caption{Estimated selection probabilities in the presence of outliers based on the maximum--likelihood estimator and on the robust estimator from Cantoni and Ronchetti (2001). The results are based on $500$ Monte Carlo simulations and the bootstrap is based on $B=50$ replications. The data has 8 moderately outlying points.} \label{table 2}
\begin{center}
\begin{tabular}{ccc|ccc|c}
\hline
& & &\multicolumn{3}{c|}{$\hat{\beta}_{ML}$} & $\hat{\beta}_{CR}$ \\
true $\beta^T$ & model & type & AIC & BIC &  $\hat{\alpha}^{s_{8}}_{m,n}$ &  $\hat{\alpha}^{s_{8}}_{m,n}$ \\
\hline 
$(1,0,0,0)$ & $(\beta_{1},0,0,0)$                   			& $\alpha_{0}$         & 0.41	&	0.42	&	0.94	&	0.94  \\
                 & $(\beta_{1},\beta_{2},0,0)$                   	& $\mathcal{A}_c$	& 0.12	&	0.12	&	0.02	&	0.02  \\
                 & $(\beta_{1},0,\beta_{3},0)$                   	& $\mathcal{A}_c$	& 0.07	&	0.07	&	0.02	&	0.02 \\
                 & $(\beta_{1},0,0,\beta_{4})$                   	& $\mathcal{A}_c$	& 0.25	&	0.24	&	0.03	&	0.02 \\
                 & $(\beta_{1},\beta_{2},\beta_{3},0)$             	& $\mathcal{A}_c$	& 0.04	&	0.04	&	0.00	&	0.00  \\
                 & $(\beta_{1},\beta_{2},0,\beta_{4})$             	& $\mathcal{A}_c$	& 0.06	&	0.05	&	0.00	&	0.00 \\
                 & $(\beta_{1},0,\beta_{3},\beta_{4})$             	& $\mathcal{A}_c$	& 0.05	&	0.05	&	0.00	&	0.00  \\
                 & $(\beta_{1},\beta_{2},\beta_{3},\beta_{4})$  	& $\mathcal{A}_c$	& 0.01	&	0.01	&	0.00	&	0.00  \\
\hline
$(-1,2,0,0)$ & $(\beta_{1},0,0,0)$                   			& --                      	& 0.00	&	0.00	&	0.00	&	0.00  \\
                 & $(\beta_{1},\beta_{2},0,0)$                   	& $\alpha_{0}$         & 0.01	&	0.01	&	0.66	&	0.78  \\
                 & $(\beta_{1},0,\beta_{3},0)$                   	& --                      	& 0.00	&	0.00	&	0.00	&	0.00 \\
                 & $(\beta_{1},0,0,\beta_{4})$                   	& --                      	& 0.00	&	0.00	&	0.00	&	0.00 \\
                 & $(\beta_{1},\beta_{2},\beta_{3},0)$             	& $\mathcal{A}_c$	& 0.00	&	0.00	&	0.01	&	0.02  \\
                 & $(\beta_{1},\beta_{2},0,\beta_{4})$             	& $\mathcal{A}_c$	& 0.79	&	0.80	&	0.33	&	0.20 \\
                 & $(\beta_{1},0,\beta_{3},\beta_{4})$             	& --	      			& 0.00	&	0.00	&	0.00	&	0.00  \\
                 & $(\beta_{1},\beta_{2},\beta_{3},\beta_{4})$  	& $\mathcal{A}_c$	& 0.20	&	0.19	&	0.01	&	0.01  \\
\hline
$(-1,1,1,0)$ & $(\beta_{1},0,0,0)$                   			& --                      	& 0.00	&	0.00	&	0.00	&	0.00 \\
                 & $(\beta_{1},\beta_{2},0,0)$                   	& --                      	& 0.00	&	0.00	&	0.01	&	0.01  \\
                 & $(\beta_{1},0,\beta_{3},0)$                   	& --                      	& 0.00	&	0.00	&	0.02	&	0.07  \\
                 & $(\beta_{1},0,0,\beta_{4})$                   	& --                      	& 0.00	&	0.00	&	0.00	&	0.00  \\
                 & $(\beta_{1},\beta_{2},\beta_{3},0)$             	& $\alpha_{0}$      	& 0.00	&	0.00	&	0.55	&	0.73 \\
                 & $(\beta_{1},\beta_{2},0,\beta_{4})$             	& --                      	& 0.00	&	0.00	&	0.04	&	0.02 \\
                 & $(\beta_{1},0,\beta_{3},\beta_{4})$             	& --                      	& 0.00	&	0.00	&	0.05	&	0.03 \\
                 & $(\beta_{1},\beta_{2},\beta_{3},\beta_{4})$  	& $\mathcal{A}_c$	& 0.99	&	0.99	&	0.34	&	0.13  \\
\end{tabular}
\end{center}
\end{table}
\clearpage

\newpage
\begin{table}[ht]
\caption{Estimated selection probabilities in the presence of outliers based on the maximum--likelihood estimator and on the robust estimator from Cantoni and Ronchetti (2001). The results are based on $500$ Monte Carlo simulations and the bootstrap is based on $B=50$ replications. The data has 2 strongly outlying points.} \label{table 3}
\begin{center}
\begin{tabular}{ccc|ccc|c}
\hline
& & &\multicolumn{3}{c|}{$\hat{\beta}_{ML}$} & $\hat{\beta}_{CR}$ \\
true $\beta^T$ & model & type & AIC & BIC &  $\hat{\alpha}^{s_{8}}_{m,n}$ &  $\hat{\alpha}^{s_{8}}_{m,n}$ \\
\hline
$(1,0,0,0)$ & $(\beta_{1},0,0,0)$                   			& $\alpha_{0}$         & 0.00	&	0.00	&	0.03	&	0.97  \\
                 & $(\beta_{1},\beta_{2},0,0)$                   	& $\mathcal{A}_c$	& 0.00	&	0.00	&	0.00	&	0.01  \\
                 & $(\beta_{1},0,\beta_{3},0)$                   	& $\mathcal{A}_c$	& 0.00	&	0.00	&	0.02	&	0.01 \\
                 & $(\beta_{1},0,0,\beta_{4})$                   	& $\mathcal{A}_c$	& 0.00	&	0.00	&	0.04	&	0.00 \\
                 & $(\beta_{1},\beta_{2},\beta_{3},0)$             	& $\mathcal{A}_c$	& 0.00	&	0.00	&	0.00	&	0.00  \\
                 & $(\beta_{1},\beta_{2},0,\beta_{4})$             	& $\mathcal{A}_c$	& 0.00	&	0.00	&	0.00	&	0.00 \\
                 & $(\beta_{1},0,\beta_{3},\beta_{4})$             	& $\mathcal{A}_c$	& 0.05	&	0.05	&	0.88	&	0.00  \\
                 & $(\beta_{1},\beta_{2},\beta_{3},\beta_{4})$  	& $\mathcal{A}_c$	& 0.95	&	0.95	&	0.04	&	0.00  \\
\hline
$(-1,2,0,0)$ & $(\beta_{1},0,0,0)$                   			& --                      	& 0.00	&	0.00	&	0.00	&	0.00  \\
                 & $(\beta_{1},\beta_{2},0,0)$                   	& $\alpha_{0}$         & 0.00	&	0.00	&	0.17	&	0.99  \\
                 & $(\beta_{1},0,\beta_{3},0)$                   	& --                      	& 0.00	&	0.00	&	0.00	&	0.00 \\
                 & $(\beta_{1},0,0,\beta_{4})$                   	& --                      	& 0.00	&	0.00	&	0.00	&	0.00 \\
                 & $(\beta_{1},\beta_{2},\beta_{3},0)$             	& $\mathcal{A}_c$	& 0.00	&	0.00	&	0.04	&	0.01  \\
                 & $(\beta_{1},\beta_{2},0,\beta_{4})$             	& $\mathcal{A}_c$	& 0.00	&	0.00	&	0.15	&	0.00 \\
                 & $(\beta_{1},0,\beta_{3},\beta_{4})$             	& --	      			& 0.00	&	0.00	&	0.00	&	0.00  \\
                 & $(\beta_{1},\beta_{2},\beta_{3},\beta_{4})$  	& $\mathcal{A}_c$	& 1.00	&	1.00	&	0.63	&	0.00  \\
\hline
$(-1,1,1,0)$ & $(\beta_{1},0,0,0)$                   		& --                      	& 0.00	&	0.00	&	0.00	&	0.01 \\
                 & $(\beta_{1},\beta_{2},0,0)$                   	& --                      	& 0.00	&	0.00	&	0.05	&	0.06  \\
                 & $(\beta_{1},0,\beta_{3},0)$                   	& --                      	& 0.00	&	0.00	&	0.00	&	0.22  \\
                 & $(\beta_{1},0,0,\beta_{4})$                   	& --                      	& 0.00	&	0.00	&	0.04	&	0.00  \\
                 & $(\beta_{1},\beta_{2},\beta_{3},0)$             	& $\alpha_{0}$      	& 0.00	&	0.00	&	0.00	&	0.71 \\
                 & $(\beta_{1},\beta_{2},0,\beta_{4})$             	& --                      	& 0.02	&	0.02	&	0.88	&	0.00 \\
                 & $(\beta_{1},0,\beta_{3},\beta_{4})$             	& --                      	& 0.00	&	0.00	&	0.00	&	0.00 \\
                 & $(\beta_{1},\beta_{2},\beta_{3},\beta_{4})$  	& $\mathcal{A}_c$	& 0.98	&	0.98	&	0.03	&	0.00  \\
\end{tabular}
\end{center}
\end{table}
\clearpage

\newpage
\begin{table}[ht]
\caption{Estimated best model for the Lindenmayer et al. (1990, 1991)  data using a range of model selection procedures.} \label{table 4}
\begin{center}
\begin{tabular}{c|c|l}
\hline
selection criterion & $\hat{\beta}$ & selected variables in the best model \\
\hline
$\hat{\alpha}$  & $\hat{\beta}_{CR}$ & stags, habitat \\
$\hat{\alpha}$  & $\hat{\beta}_{ML}$ & stags, habitat \\
AIC  & $\hat{\beta}_{ML}$ & stags, bark, acacia, habitat, aspect \\
BIC  & $\hat{\beta}_{ML}$ & stags, bark, acacia, aspect \\
$p$-value forward stepwise &$\hat{\beta}_{CR}$ & stags, bark, acacia, habitat, aspect \\
$p$-value forward stepwise& $\hat{\beta}_{ML}$ & stags, bark, acacia, habitat, aspect \\
\end{tabular}
\end{center}
\end{table}
\clearpage

\end{document}